\documentclass[aps,superscriptaddress,showpacs,preprint,nofootinbib]{revtex4}
\usepackage{amssymb}
\usepackage{amsmath}
\usepackage{amsthm}
\newcommand\scri{\mathcal{I}}

\newtheorem{thm}{Theorem}
\newtheorem{Def}{Definition}
\newtheorem{cor}[thm]{Corollary}

\begin{document}

\title{Non-existence of Asymptotically Flat Geons in $2+1$ Gravity}
\author{Kory A. Stevens}
\affiliation{Department of Physics and Astronomy, University of British Columbia,
Vancouver, British Columbia, Canada \ V6T 1Z1}
\author{Kristin Schleich}
\author{Donald M. Witt}
\affiliation{Department of Physics and Astronomy, University of British Columbia,
Vancouver, British Columbia, Canada \ V6T 1Z1}
\affiliation{Perimeter Institute for Theoretical Physics, 31 Caroline Street North, Waterloo,
Ontario, Canada N2L 2Y5}
\date{\today}

\begin{abstract}
Geons, small topological structures that exhibit particle properties such as charge and angular momentum without the presence of matter sources, have been extensively discussed in  $3+1$-dimensional general relativity. Given the recent renewal of interest in $2+1$ gravity,  it is natural to ask whether or not the notion of geons extends to three dimensions. We prove here that, in contrast to the $3+1$-dimensional case,  there are no $2+1$-dimensional asymptotically flat solutions of the vacuum Einstein  or Einstein-Maxwell equations containing geons.  In contrast, $2+1$-dimensional asymptotically anti-de Sitter spacetimes can indeed contain geons; however, the geons are always hidden behind a single black hole horizon. We also prove sufficient conditions for the non-existence of $2+1$-dimensional asymptotically flat geon-containing solutions. \\
 
\normalsize  
\noindent {\it This article is dedicated to Rafael Sorkin, whose encouragement is responsible for one of us (DW) pursuing research in physics and whose work on topology and quantum gravity has inspired all of us.}
\end{abstract}
\maketitle

\section{Introduction}

Over half a century ago, Wheeler proposed that geons, small isolated gravitational structures without horizons in general relativity, could exhibit particle-like properties such as charge and angular momentum without the presence of matter sources \cite{Wheeler:1955zz}.  This work, further elaborated by Wheeler and his collaborators in  \cite{Wheeler:1957mu,Misner:1957mt,Brill:1957,Ernst:1957,Brill:1964} provided an intriguing alternative  to the introduction of fundamental matter fields in quantum theories that include gravity. The early study of geons  emphasized the analysis  of 
asymptotically flat spacetimes containing soliton-like gravitating electromagnetic fields with simple topology ${\mathbb R}^3$; however initial data for geons with nontrivial topology was also exhibited, mainly in the context of the production of nontrivial electric charge. Though conceptually appealing, technical problems were discovered that reduced interest, namely that asymptotically flat geons on ${\mathbb R}^3$ with no horizons were unstable on short timescales; the solitonic-like configurations tended either to collapse to a black hole or disperse. In addition, early topological geons exhibited the unpleasing property that both electric and magnetic charges could be produced by threading the  topological structure with the appropriate electromagnetic field. Thus these geons yielded no selection rule for the observed absence of magnetic monopoles. However in 1979, Sorkin showed that a geon with the topology of a non-orientable handle did not exhibit this flaw \cite{Sorkin:1979ja}; it produced only one kind of monopole charge, which could be taken as the electric charge. Soon after, Friedman and Sorkin gave an interesting formal argument that the inclusion of topological geons in  4-dimensional quantum 
gravity resulted in  spin 1/2 quantum states with no fermionic matter sources \cite{Friedman:1980st,Friedman:1982du}. These two arguments renewed interest in the role of geons, in particular those with nontrivial topology, as quantum particles in 4-dimensional quantum gravity \cite{Sorkin:1985bg}. Henceforth, the term geon will refer to ones with nontrivial topological structure.

Further work elaborated key features of geons in classical general relativity.
The detailed conditions for the formal
existence of spin 1/2 states from geons and 
interesting ties to the topology of 3-manifolds were given in a series of papers by Friedman and
Witt  \cite{Friedman:1983ft, Witt:1986ef, Friedman:1986ze, Friedman:1988he}.\footnote{These results also yielded counter-examples to some conjectures in $3$-dimensional 
topology.} In addition, Witt showed that physically reasonable initial data sets for the Einstein equations can be constructed on all smooth 3-manifolds \cite{Witt:1986ng}. Thus any quantum geon in 4-dimensional quantum gravity will yield a classical geon  in the correspondence limit. This analysis of geons in classical relativity laid the foundation for further formal study of geons in 4-dimensional quantum gravity, in particular the study of the applicability of spin-statistics to quantum geons \cite{Sorkin:1985bg,Aneziris:1988xz,Aneziris:1989cr,Balachandran:1990wr,Sorkin:1996yt,Dowker:2000zy}. These formal results on spin 1/2 quantum states of geons and their spin-statistics are generic in the sense that they apply in any theory of quantum gravity in which the rotation operators can be defined at  asymptotic infinity. Nonetheless, no concrete application of these results in a quantum setting can be investigated as no complete theory of quantum gravity  in 4-dimensions is known.  Therefore, order to gain a
better understanding of the physics of quantum geons,  it would be useful to have an analogous, well defined quantum arena to be able to investigate their properties in.

 A natural starting point for such an investigation is  $2+1$-dimensional quantum gravity.
Twenty years ago, it was pointed out by Witten that $2+1$-dimensional quantum gravity  was exactly solvable \cite{Witten:1988hc}.  By working on the space of solutions, quantum amplitudes for $2+1$ gravity with closed spatial topologies could be formulated in terms of finite dimensional quantum mechanics.
Thus problems such as the nonrenormalizability of quantum gravity in more than two dimensions associated with traditional approaches based on quantization on infinite dimensional configuration space could be avoided.
This result led to a renaissance of research in Chern-Simons theories and topological field theories in three and more dimensions, mainly concentrating on the flat, closed universe case (See \cite{Carlip:1998uc} and \cite{Carlip:2004ba} for an overview of this large set of results).
Indeed, this work spurred the investigation of formal properties of $2+1$-dimensional geons such as spin \cite{Samuel:1993ua} and a generalized spin 
statistics theorem for the geons in $2+1$-dimensional quantum gravity \cite{Balachandran:2000kv} based on the assumption that quantization of closed, flat universes could be extended to the open case.

Recently, there has been  renewed interest in $2+1$-dimensional gravity; Witten 
 pointed out certain issues are more subtle in the asymptotically anti-de Sitter case and should be revisited, a case not considered in detail in the first era of work \cite{Witten:2007kt}. This suggestion has resulted in a resurgence of work in $2+1$-dimensional quantum gravity in the context of $2+1$-dimensional anti-de Sitter spacetimes and related $2+1$-dimensional theories such as topologically massive gravity \cite{Carlip:2008jk,Carlip:2008eq} and chiral gravity \cite{Li:2008dq}. These $2+1$-dimensional quantum gravity models explicitly  exhibit asymptotic regions. Therefore they provide a natural  forum in which it may be possible to extend formal results on geons and other such structures  more rigorously to the quantum realm.  
 
 Identification of classical $2+1$-dimensional geons is a key starting point for the study of quantum geons; it allows identification of formal results on the spin and charge of geons to their quantum counterparts.  Furthermore, if such classical geons exist, then so should their quantum counterparts by the correspondence principle. This paper will address this key question; are there geons with nontrivial topology in classical $2+1$ gravity? We will prove that, in contrast to the $3+1$-dimensional case, there are no asymptotically flat solutions to the $2+1$-dimensional vacuum Einstein or Einstein-Maxwell equations that contain geons. Hence it is unclear how to relate formal results on $2+1$ geons to quantum ones for asymptotically flat spacetimes. On the other hand, geons and other spacetimes with nontrivial topology do exist in $2+1$-dimensional asymptotically AdS spacetimes \cite{Brill:1995jv,Aminneborg:1997pz,Aminneborg:1998si}. Hence these spacetimes would seem to provide a more natural starting point for the study of $2+1$ geons.

In Section \ref{prelim}, we summarize basic definitions needed in the proof of our results. In Section \ref{sec3}, we prove a contradiction to the existence of an outer trapped surface for  asymptotically flat $2+1$-dimensional solutions to the vacuum Einstein and Einstein-Maxwell equations and provide sufficient conditions for the non-existence of an outer trapped surface for spacetimes with zero cosmological constant that satisfy the dominant energy condition.  We begin with a proof of a result of  Ida \cite{Ida:2000jh} on the nonexistence of trapped surfaces in spacetimes with positive cosmological constant. We then fix a gap in the extension of Ida's result to the case of zero cosmological constant by first proving that there exist no trapped surfaces in this case if the spacetime is analytic. We conclude by showing that $2+1$-dimensional Einstein spaces are indeed analytic. In Section \ref{sec4} we prove our main result.  First, topological rigidity,
 a consequence of  topological censorship  \cite{Friedman:1993ty,Galloway:1999bp,Galloway:1999br}, is used to 
prove the existence of a horizon for both the asymptotically flat and asymptotically anti-de Sitter geons. 
We then prove that asymptotically flat geon spacetimes must contain an outer trapped surface. We then use the results of Section \ref{sec3} to prove a contradiction to the existence of a horizon for asymptotically flat spacetimes of nontrivial topology for the vacuum and Einstein-Maxwell cases. Finally we note that, in contrast,
asymptotically AdS spacetimes admit an outer trapped surface and thus there is no contradiction.
Hence geons exist in $2+1$-dimensional asymptotically AdS spacetimes. 

\section{Preliminaries}\label{prelim}

We begin by recalling some basic definitions from general relativity needed in the statement and proof of the theorems.

A spacetime satisfies the {\it null energy condition} (NEC) if $R_{ab} W^a W^b \geq 0$ 
for all null $W^a$, the {\it weak energy condition} (WEC) if $T_{ab} W^a W^b \geq 0$ 
for all timelike $W^a$  and the {\it dominant energy condition} (DEC) if the WEC holds and
$T_{ab} W^b T^a_{c} W^c \leq 0 $ \cite{he}. Notice the weakest of these energy conditions 
is the NEC condition. If the DEC is satisfied,  it directly follows that the WEC is satisfied. If the WEC is satisfied, then so is the NEC by a continuity argument. 

 The {\it universal covering space} or {\it universal covering manifold} of $M$ can be constructed in the following way: Pick a point $x_0\in
M$ and consider the set of smooth paths $P=\{c:[0,1]\rightarrow M|c(0)=x_0\}$.  A
projection map $\pi:P\rightarrow M$ is defined by  $\pi(c(t))=c(1)$.  Let ${\cal
M}$ be $P$ modulo the equivalence relation, $c_1\sim c_2$ if and only if
$c_1(1)=c_2(1)$ and $c_1$ is homotopic to $c_2$ with endpoints fixed. The
projection map $\pi$ is then well defined and smooth as a map $\pi:  {\cal
M}\rightarrow M$. By construction, the universal covering manifold ${\cal M}$ is simply connected.

Next, a $n+1$-dimensional spacetime $( M, g_{ab})$ is {\it asymptotically flat} (AF) if it can be
conformally included into a spacetime-with-boundary
$M'=  M \cup \scri$, with metric $g_{ab}'$, such that (a)
 for some conformal
factor $\Omega\in C^{1}(M')$, $g_{ab}' = \Omega^2 g_{ab}$ on $M$  and $\Omega$ vanishes on $\scri$ but has null gradient which is nonvanishing  pointwise on $\scri$. (b) The boundary $\partial 
M' = M'\setminus M= \scri$ is a disjoint union of past and future parts $\scri^+ \cup\scri^-$, each having topology $S^{n-1}\times \mathbb R$ with  complete null generators. 

A $n+1$-dimensional spacetime $(M, g_{ab})$ is {\it asymptotically locally anti-de\! Sitter} (ALADS) if it can be
conformally included into a spacetime-with-boundary
$M'=  M \cup \scri$, with metric $g_{ab}'$, such that $\partial 
M' = \scri$ is timelike ({\it i.e.}, is a Lorentzian hypersurface 
in the induced metric) and $M = M' \setminus \scri$. 
The conformal
factor $\Omega\in C^{1}(M')$ satisfies (a) $\Omega > 0$ and $g_{ab}' = \Omega^2 g_{ab}$  on $M$, and (b) $\Omega = 0$ and $d \Omega \ne 0$ pointwise on $\scri$. We permit $\scri$ to have multiple components.

A spacetime $M$, possibly with boundary, is {\it globally hyperbolic} if it is strongly causal and the 
sets $J^+(p,{M})\cap J^-(q,M)$ are compact for all 
$p,q\in M$. \footnote{The timelike future (causal future)
of a set $S$ relative to  $U$, $I^+(S,U)$ ($J^+(S,U)$), is the 
set of all points that can be reached from $S$ by a future directed  
timelike curve (causal curve) in $U$.  The interchange of the past with future in the previous
definition yields $I^-(S,U)$ ($J^-(S,U)$). }

This definition is a generalization of that of a globally hyperbolic spacetime without boundary and is satisfied by ALADS spacetimes.\footnote{ In fact, it is that used in the proof of topological censorship in ALADS spacetimes \cite{Galloway:1999bp}.}  Also, note that the Penrose compactification of an AF spacetime (which is itself globally hyperbolic by the usual definition) is globally hyperbolic in this general sense.

A   {\it Cauchy surface} $V$
is a spacelike hypersurface 
such that every non-spacelike curve intersects this surface exactly once. Note $V$ for a manifold with boundary $\scri$ will 
have boundary 
on $\scri$. A {\it partial Cauchy surface} is a surface that satisfies the weaker condition that each non-spacelike curve intersects the surface at most once. 

The {\it domain of outer communications} (DOC) is the portion of a spacetime  ${ M}$ 
which is exterior to event horizons. Precisely ${D} = I^-(\scri^+_0)\cap I^+(\scri^-_0)$ for a connected component $\scri_0$ for an AF spacetime and
${D} = I^-(\scri_0)\cap I^+(\scri_0)$ for an ALADS spacetime. Intuitively, the DOC is the subset of ${ M}$ that
is in causal contact with $\scri$. Note that $D$ is the interior of an $(n+1)$-dimensional 
spacetime-with-boundary ${D}' = {D }\cup \scri$ and that $D'$ is itself a globally 
hyperbolic spacetime with boundary.

An {\it event horizon} is the boundary of the  DOC. More specifically, a {\it future event horizon} is 
the boundary of the causal past of a connected component of the boundary at infinity, $\scri_0$, $ \dot J^{-}(\scri_0, M')$, a {\it past event horizon}
is the boundary of the causal future of $\scri_0$, $ \dot J^{+}(\scri_0, M')$ and the {\it event horizon} is the union of future and past event horizons.

A closed orientable spacelike surface in the Cauchy surface $V$ is an {\it outer trapped surface} if the expansion of outgoing null geodesics orthogonal to it is nonpositive, $\theta \leq 0$. The case $\theta=0$ is  a {\it marginally outer trapped surface}, also termed an {\it apparent horizon}. 
Note that in 2+1 dimensions, an outer trapped surface is in fact a curve with the topology of a circle; however, we will adopt the convention of referring to it as a surface in keeping with the nomenclature of higher dimensions.

\section{Existence and Nonexistence of  Outer Trapped Surfaces in 2+1 Gravity}\label{sec3}

We begin with a set of theorems on the existence of outer trapped surfaces in AF and ALADS spacetimes in $2+1$ gravity. These theorems do not depend on the topology of the Cauchy surface; however they play an essential role in the analysis of geons in the next section. The first theorem demonstrates that a necessary condition for an outer trapped surface  is nonpositive cosmological constant. This result was first obtained by  Ida \cite{Ida:2000jh}. Our proof, using the triad notation of \cite{Ashtekar:2002qc},  parallels that of  \cite{Ida:2000jh}; both are based on the approach used in proving the topology of a marginally outer trapped surface by Hawking \cite{Hawking:1973} and that of the horizon in stationary spacetimes \cite{Hawking:1971vc,he}. (See \cite{Woolgar:1999yi} for a more readily available outline of the Hawking proof.) We then sharpen this result and fix a gap in its application  to the zero cosmological constant case, pointed out in the 4-dimensional case by Galloway \cite{Galloway:1994a} by proving that there are no outer trapped surfaces in AF spacetimes that are analytic. Finally we prove that all $2+1$-dimensional spacetimes with Einstein metric are analytic. 
\begin{thm}\label{trapsur}  
Let $M^{2+1}$ be a globally hyperbolic spacetime that satisfies the Einstein equations with cosmological constant,
$R_{ab}-\frac 12 R g_{ab} +\Lambda g_{ab} =8 \pi T_{ab}$, where
$T_{ab}$ obeys the DEC. If $M^{2+1}$ contains an outer trapped surface, 
and the Cauchy surface containing it is not itself entirely trapped, then
$\Lambda \leq 0$.
\end{thm} 

\begin{proof} Let $V$ be a Cauchy surface in $M^{2+1}$ that  contains an outer trapped surface. This surface either  is itself a marginally outer trapped surface or lies in a trapped region. If the trapped region is  not all of $V$,  it  must have an outer boundary and it follows that this outer boundary is the outermost marginally outer trapped surface in $V$. Denote this outermost marginally outer trapped surface by $\mathcal{S}$. Let $l$ be the future directed null vector tangent to a null congruence orthogonal to $\mathcal{S}$, $n$ be the other future directed null vector orthogonal to $\mathcal{S}$ and $m$ be a spatial vector tangent to $\mathcal{S}$.  Normalize $l$, $n$ and $m$ such that 
\begin{align*} l \cdot n = -1 && l\cdot m = n\cdot m = 0&&m\cdot m = 1 \ .\end{align*}
The metric can be decomposed in terms of this null triad;
\begin{align*} g_{ab} = -l_a n_b -l_bn_a + m_am_b&&  g^{ab} = -l^a n^b -l^bn^a + m^am^b \ .\end{align*}
The covariant derivatives of $l$, $n$ and $m$ can be written in terms of the spin coefficients c.f.  \cite{Ashtekar:2002qc}:
\begin{align}
\nabla_al_b &= -\epsilon n_al_b + \kappa_{NP} n_a m_b - \gamma l_al_b + \tau l_a m_b + \alpha m_a l_b - \rho m_a m_b\nonumber\\
\nabla_a n_b &= \epsilon n_a n_b - \pi n_a m_b + \gamma l_a n_b - \nu l_a m_b - \alpha m_a n_b + \mu m_a m_b\nonumber\\
\nabla_a m_b &= \kappa_{NP} n_a n_b - \pi n_a l_b + \tau l_a n_b - \nu l_al_b - \rho m_a n_b + \mu m_a l_b\ .\label{spincoeff}
\end{align}
As the spacetime is $2+1$-dimensional, there is no shear or vorticity; the expansion  $$\theta= m^am^b\nabla_al_b= -\rho$$
 completely characterizes the projection into $\mathcal{S}$ of the evolution of this congruence.
 
Consider deforming $\mathcal{S}$  outwards in $V$ a distance $s$ along the spatial vector $v^a=e^y(l^a - n^a)$ in $V$.
\footnote{Note that $l^a$ and $n^a$ can be scaled, $l^a \to e^\phi l^a$ and $n^a\to e^{-\phi} n^a$, such that $l^a - n^a$ lies in $V$ and is proportional to the normal of $\mathcal{S}$ in $V$.}
Call this surface $\mathcal{S}(s)$. The scaling factor $y$ is initially arbitrary; it will be chosen to take on a convenient value later. $l^a$ and $n^a$ will be orthogonal to $\mathcal{S}(s)$ if 
\begin{align*}
l^am^b\nabla_b v_a = l^a v^b\nabla_b m_a &&\hbox{\textrm and}&& n^am^b\nabla_bv_a = n^a v^b\nabla_b m_a
\end{align*}
or, in terms of the spin coefficients,
\begin{align}
m^b\nabla_b y=\alpha + \kappa_{NP} - \tau= -\alpha+\pi-\nu\ .\label{ortho}\end{align}

A straightforward computation yields
\begin{align}l^a\nabla_a\theta &=  \kappa_{NP}(\tau - \pi +2\alpha) -\rho^2 -\rho\epsilon- m^a\nabla_a\kappa_{NP}-R_{ab}l^al^b\nonumber\\
-n^a\nabla_a \theta &=-\tau^2-(\mu-\gamma)\rho + \nu\kappa_{NP}+ m^a\nabla_a \tau - 
R_{abcd}n^al^bn^cl^d+R_{ab}n^al^b\label{nablaexp}
 \end{align}
 where the curvature convention used is $\nabla_a\nabla_b t_c - \nabla_b\nabla_a t_c = R_{abc}^{ \ \ \ d} t_d$.
In three dimensions, the Riemann curvature is determined by the Ricci curvature:
\begin{equation}\label{3Ricci}
R_{a b c d} = g_{ac} R_{b d} + g_{b d} R_{a c} - g_{ad} R_{b c} - g_{b c} R_{a d} +\frac 12 (g_{a d} g_{b c} - g_{a c} g_{b d}) R \ .
\end{equation}
Therefore
\begin{equation} - 
R_{abcd}n^al^bn^cl^d+R_{ab}n^al^b = -R_{ab}n^al^b+ \frac 12 R= -8\pi T_{ab}n^al^b -\Lambda\label{curve}
\end{equation}
upon imposing the Einstein equations.
It directly follows from (\ref{nablaexp}) after substitution from  (\ref{ortho}) and (\ref{curve}) that
\begin{equation} e^{-y}v^a\nabla_a \theta = -(\tau - \kappa_{NP})^2  + m^a\nabla_a(\tau-\kappa_{NP}) - \rho(\rho+\epsilon+ \mu - \gamma) -R_{ab}l^al^b-8\pi T_{ab}n^al^b - \Lambda \ . 
\label{thetapush}
\end{equation}
The orthogonality condition (\ref{ortho}) also yields $ \tau - \kappa_{NP}= \alpha-m^a\nabla_a y $. Using this relation in the divergence of $\tau - \kappa_{NP}$ in (\ref{thetapush}) and evaluating on $\mathcal{S}$ yields
\begin{equation} e^{-y}v^a\nabla_a \theta\bigg|_{\mathcal S} = -(\tau - \kappa_{NP})^2 - \Lambda  + m^a\nabla_a(\alpha-m^b\nabla_b y) -R_{ab}l^al^b-8\pi T_{ab}n^al^b\label{thetapush2} \end{equation}
Now, denoting $m^a\nabla_a =\frac {d\ }{dx}$, note that $-\frac {d^2y}{dx^2} +f(x) = 0$ has a solution if and only if $\int_{\mathcal{S}} f(x) = 0$. Defining
$$ c =  \int_{\mathcal{S}} R_{ab}l^al^b+8\pi T_{ab}n^al^b$$
one can use this fact to choose  the scaling $y$ such that the last three terms in (\ref{thetapush2}) are constant and
\begin{equation}\label{needed} \frac {d\theta}{ds} \bigg|_{\mathcal S}=v^a\nabla_a \theta = -e^y\left((\tau - \kappa_{NP})^2 + \Lambda  +c\right)\ .\end{equation}
The constant $c$ is nonnegative if the DEC holds; the DEC implies the NEC as well as that  $ T_{ab}l^b$ is a future pointing non-spacelike vector. Therefore  $ R_{ab}l^al^b$, $ T_{ab}n^al^b$ and consequently $c$ are nonnegative. Thus the right hand side of (\ref{needed}) is negative if $\Lambda >0$. This implies that there exists a trapped surface outside of $\mathcal{S}$, in contradiction to the assumption that it was the outermost marginally outer trapped surface. Hence $\Lambda\leq 0$. \end{proof}

Theorem \ref{trapsur} does not directly use the asymptotic properties of  the spacetime, although it does assume that the notion of an outer trapped surface is well defined, which implies that there is some asymptotic region with respect to which the outer direction can be defined. In addition, asymptotics can be used to show that  the trapped region does not consist of all of $V$. In particular,  AF and ALADS spacetimes have the property that large spatial circles near the boundary at infinity are untrapped. Hence, the asymptotics of these spacetimes guarantee that the trapped region is not all of $V$ and hence that there is a marginally outer trapped surface.

Theorem \ref{trapsur}  directly tells us that there can be no outer trapped surface in a spacetime with $\Lambda >0$.\footnote{This conclusion is, of course, subject to having an asymptotic region in a $\Lambda>0$ spacetime with respect to which a suitable definition of the outer direction can be defined.} If $\Lambda <0$, it implies that a spacetime can indeed have trapped surfaces: the BTZ black hole  \cite{Banados:1992wn} is an explicit example of such a spacetime. 

The case $\Lambda=0$ is more complicated. If the matter source is such that $c<0$ on $\mathcal{S}$, indicating that matter is crossing the apparent horizon, then again one immediately concludes that there can be no trapped surfaces in the spacetime. For example, spacetimes containing electromagnetic fields will have $c<0$; hence there can be no outer trapped surfaces in $2+1$-dimensional solutions of the Einstein-Maxwell equations. This conclusion clearly generalizes to other long range fields satisfying the DEC.  However, Theorem \ref{trapsur} leaves open the possibility that a trapped surface can exist in a spacetime with $\Lambda = 0$. Such a spacetime would necessarily have vanishing stress energy projection $T_{ab}l^a(n^b+l^b)$. In particular, the key case of vacuum spacetime is included in these possibility.

We begin by ruling out the possibility of an outer trapped surface  in analytic spacetimes with $\Lambda=0$. 
\begin{cor} \label{noanalytic} If $M^{2+1}$, satisfying the conditions of Theorem \ref{trapsur}, is an analytic spacetime with zero cosmological constant, then there are no outer trapped surfaces.
\end{cor}

\begin{proof}
If $c>0$ and/or $\tau-\kappa_{NP}\neq 0$ on $\mathcal{S}$, the conclusion follows immediately from Theorem \ref{trapsur}. Therefore, assume these terms vanish at all points on $\mathcal{S}$. By analyticity,
$\theta$ is  given by its Taylor expansion, \begin{equation*}
\theta(s) = \sum_{n=0}^\infty c_ns^n 
\end{equation*}  and $c_0$ and $c_1$ vanish everywhere on $\mathcal{S}$. Next, observe that
\begin{equation}\label{limit}
c_2=\lim_{s\rightarrow 0^{+}} \frac{\theta} { s^2}= \lim_{s\rightarrow 0^{+}}\frac{d\theta/ds}{2 s} 
\end{equation}
By (\ref{thetapush}) and appropriate choice of $y$ (\ref{limit}) is equivalent to
\begin{equation}
c_2 = \lim_{s\rightarrow 0^{+}} e^y \frac {\theta}{2s}(-\theta + \epsilon - \mu-\gamma) - \frac {e^y}{2s}( (\tau-\kappa_{NP})^2  +c) \ .\end{equation}
The spin coefficients, $c$ and $y$ are analytic functions of $s$ by the assumption of analyticity.
Taking the limit, as the leading term in $\theta$ is $c_2 s^2$,  the first term on the right hand side vanishes, and the second term is manifestly nonpositive. It follows that $c_2\leq 0$ everywhere on $\mathcal{S}$. If $c_2<0$ for any point, there is an outer trapped surface in some neighborhood outside of  $\mathcal{S}$, in contradiction to initial assumption. If $c_2 =0$ at all points on $\mathcal{S}$, then  iteration of the above argument yields the result that if the coefficients $c_{n}$ vanish on $\mathcal{S}$ for all $n\leq k-1$, 
\begin{equation}
c_k =\lim_{s\rightarrow 0^{+}}\frac{d\theta/ds}{k s^{k-1}}=\lim_{s\rightarrow 0^{+}} e^y \frac {\theta}{ks^{k-1}}(-\theta + \epsilon - \mu-\gamma) - \frac {e^y}{ks^{k-1}}( (\tau-\kappa_{NP})^2  +c)\leq 0 \ .\end{equation}
If $c_k<0$ at some point on $\mathcal{S}$, it again follows that  there is an outer trapped surface in some neighborhood of it, in contradiction to initial assumption. If all coefficients vanish identically, then it follows that $\theta$ must be zero somewhere for $s>0$; that is, there exists a marginally outer trapped surface outside of $\mathcal{S}$, again in  contradiction to initial assumption.
\end{proof}

We will next consider the conditions under which the spacetime metric is analytic.  First, a well known result is that static vacuum spacetimes are analytic  \cite{Muller:1970}. This result was extended to the stationary case in \cite{Muller:1971} and explicitly worked out for the case of Einstein-Maxwell gravity in \cite{Tod:2007tb}. These results were proven in 4 dimensions, but can be  generalized to any dimension. As the future horizon and apparent horizon coincide for stationary spacetimes, one might conclude that stationary spacetimes satisfy the conditions of  Corollary \ref{noanalytic}. However, note that the methods used in the proof of the results \cite{Muller:1970,Muller:1971,Tod:2007tb} explicitly exclude the horizon. Thus, strictly speaking, these results  only prove analyticity outside of it. Hence, although all known stationary solutions are in fact analytic on the horizon as well, there is a potential issue in assuming analyticity at the horizon as needed for Corollary \ref{noanalytic}. 

However, in 3 dimensions, one can prove a stronger result that covers this gap for the vacuum case; $2+1$-dimensional vacuum spacetimes are analytic everywhere. This result does not assume staticity or stationarity.

\begin{thm}\label{analy}  
Any $2+1$-dimensional spacetime $M^{2+1}$ with Einstein metric is analytic.
\end{thm} 
\begin{proof} 
An Einstein metric in $2+1$ gravity satisfies
\begin{equation*}
R_{ab}={2 \Lambda }  g_{ab} \ .
\end{equation*}
where $\Lambda$ is the cosmological constant.
The $3$-dimensional Riemann tensor (\ref{3Ricci}) is then
\begin{equation*}
R_{a bcd} 
= {\Lambda} (g_{b d}g_{a c} - g_{b c}g_{a d})  
 \, .
\end{equation*}
\noindent Note that the Riemann curvature tensor is parallel because 
$\nabla _a {R_{b c d e}}=0$.
Next, recall that the sectional curvature is defined by
\begin{equation*}
K(X,Y)={\frac{R_{abcd}X^a Y^b X^c Y^d }
{(g_{ac }g_{b d } -
 g_{a d}g_{bc})X^a Y^bX^c Y^d}} \, .
\end{equation*}
The sectional curvature for an Einstein metric in 3 dimensions is thus simply
\begin{equation*}
K(X,Y) ={ \Lambda} \, .
\end{equation*}

\noindent Therefore, a  $2+1$-dimensional spacetime with Einstein metric  is a spacetime of constant sectional curvature, that is  a {\it local
space form}. In general, one can pick local coordinates  at each point  such that the metric can be
written as
 \begin{equation}
 \label{spaceform}
ds^2= {\frac{-dT^2+du^2+dv^2} {{ 1+ {{\frac 14}{\Lambda}}  (u^2+v^2-T^2)}^2}} \, .
\end{equation} 
This metric is thus analytic in a neighborhood of the point as the conformal factor is an analytic function. 
\end{proof}

This result is entirely a consequence of  the dimensionality. In general, Einstein metrics in 4 or more spacetime dimensions have non-zero Weyl curvature. Physically,
this is due to the existence of gravitational waves.  Therefore the analogue of Theorem \ref{analy} will not hold in general for Einstein metrics in four or more spacetime dimensions. However, the same proof and conclusions are true for a particular Einstein metric in any dimension if the Weyl curvature of that metric vanishes.

Theorem \ref{analy} shows that any $2+1$-dimensional spacetime with Einstein metric is not only analytic but is also a local space form. 
Furthermore, Theorem \ref{analy} is true regardless of whether or not the spacetime is spatially compact or open. 
Moreover, it applies even if the metric is not globally hyperbolic, for example if the spacetime has 
closed timelike curves.
Finally, the spacetime in Theorem \ref{analy} is not assumed to be complete. If the spacetime is complete, 
then it must be obtained via a discrete group action on a maximal symmetric space (See, for example, \cite{wolf}).
However, if the spacetime is not complete, the global structure of the spacetime is much more complicated 
 than that given by the simple local metric (\ref{spaceform}).
This was first pointed out by Morrow-Jones and Witt for the case of positive cosmological constant in 4 dimensions  \cite{MorrowJones:1988yw, MorrowJones:1993zu}. Later interest by
mathematicians  in $2+1$-dimensional Lorentzian spacetimes \cite{mess} 
 was spurred by connections between spacetime structures, conformal structures and the work of Thurston.  

Theorem \ref{analy} shows that  the conditions of Corollary \ref{noanalytic} to Theorem \ref{trapsur} always hold for a vacuum AF spacetime. Therefore, there are no outer trapped surfaces in $2+1$-dimensional vacuum AF spacetimes. 

\section{Geon Spacetimes and Horizons 2+1 Gravity}\label{sec4}

We have established that there are no outer trapped surfaces in $2+1$-dimensional vacuum AF spacetimes. However, this fact does not itself show that there are no horizons. A standard result in causal structure is that if a globally hyperbolic AF spacetime contains an outer trapped surface, this surface must lie inside a future event horizon \cite{he}. However, the converse is not true. 
A simple illustration of this,  a slicing of Schwarzschild spacetime in which there are no outer trapped surfaces
 in the black hole region of the Cauchy surfaces was exhibited by Wald and Iyer \cite{Wald:1991}. However, one can prove, using topological censorship, that  2+1 AF or ALADS geon spacetimes must have horizons. Furthermore, these spacetimes also must contain an outer trapped surface.  These theorems, proven below, allow us to conclude our main result, the nonexistence of AF vacuum geon spacetimes.

We begin with the definition of a geon spacetime:
\begin{Def}
A globally hyperbolic spacetime $M$ 
 is a $2+1$-dimensional {\it geon} if its 
Cauchy surface $V$ has interior homeomorphic to $\Pi - \{p\}$ where $\Pi$ is a closed 2-manifold other than the 2-sphere $S^2$, and $\{p\}$ is a 
 point. 
 \end{Def}
 Note that a closed manifold is alternately termed a compact manifold without boundary.
An AF geon is a geon that also satisfies the definition of an AF spacetime. An ALADS geon similarly  also satisfies the definition of an ALADS spacetime.

To prove our next result,  we will utilize a key corollary of the topological censorship Theorem  for $2+1$-dimensional spacetimes \cite{Galloway:1999br}. In  Theorems \ref{topcen} and \ref{topcenhor} below, the energy condition used is the NEC. However, the topological censorship theorems  hold under a form of the weaker, averaged null energy condition (ANEC).\footnote{ Specifically, the form of ANEC used is that for each point $p$ in $\cal M$ near ${\scri}$ and any
future complete null
geodesic $s\to\eta(s)$ in $\cal M$ starting at $p$ with tangent $l$, $\int_0^{\infty}R_{ab}l^al^b\,ds\ge
0$.} Clearly, Theorems \ref{topcen} and \ref{topcenhor} hold under the more general energy condition, ANEC. However, the DEC implies  the NEC; therefore it  
suffices for the goals of this paper to state these theorems using the NEC.

\begin{thm}\label{topcen}  Let $D$ be the DOC of a globally hyperbolic, $2+1$-dimensional AF or ALADS spacetime  satisfying the NEC, and let $V'=V\cup\partial V$ be a Cauchy surface in the Penrose compactification of $D$. Then  
$V$ is either $B^2$ (a disk) or $I\times S^1$ (an annulus).
\end{thm} 

\begin{proof} (We provide here an extended, alternate proof to that in \cite{Galloway:1999br}.)  Let $V'$ be the 2-dimensional Cauchy surface in the Penrose compactified spacetime $D'=D\cup\scri$  with boundary at spatial infinity $\partial V =\sigma_{\infty}$. As $D'$ satisfies the conditions needed to prove topological censorship, it follows that the homomorphism of fundamental groups $i_*: \pi_1(\sigma_{\infty}) \to \pi_1(V')$ induced by inclusion is 
surjective and $D'$ is orientable if $\scri$ is orientable \cite{Galloway:1999bp,Galloway:1999br}. 
That is, the sequence $\pi_1(\sigma_{\infty}) \to \pi_1(V') \to 1$ is exact.

In 3 dimensions, each connected component of $\scri$  has topology $S^1\times {\mathbb R}$ for  AF geons and
$S^1\times {\mathbb R}$ or ${\mathbb R}^2$ for  ALADS geons. 
Therefore, the only choices for the spatial boundary topology at infinity are $S^1$, the only closed connected 1-dimensional manifold, and $\mathbb R$. First consider the case where the spatial boundary at infinity is $S^1$. Then $\pi_1(S^1) \to \pi_1(V')\to 1$ and $V'$ is orientable. 
Since $\pi_1(S^1) ={\mathbb Z}$, it follows that  ${\mathbb Z} \to \pi_1(V')\to 1$. This exact 
sequence  and basic group theory imply that  $\Pi_1(V')= {\mathbb Z}/{\rm ker\  i_*}$. Since the 
kernel of $i_*$ must be a subgroup of the integers, 
${\rm ker\  i_*}\subset {\mathbb Z}$,  it follows that ${\rm ker\  i_*} = s {\mathbb Z}$ where s is a fixed non-negative 
integer. Therefore, $\pi_1(V')= {\mathbb Z}/ s {\mathbb Z}= {\mathbb Z}_s$. Now, from the classification of 2-manifolds, $V'$ must be a closed orientable 2-manifold  minus a disjoint
union of one or more disks.  It follows from the classification of 2-manifolds and their fundamental groups that $s=0$ or $s=1$ and
that the only choices for $V$ are $B^2$ (a disk) 
or $I\times S^1$ (an annulus). Finally, consider the case where the spatial boundary at infinity topology is $\mathbb R$. It immediately follows from the fact that $\mathbb R$ is contractible and  $\pi_1({\mathbb R}) \to \pi_1(V')\to 1$ that $V'$ is also contractible, 
$\pi_1(V')= 1$. Hence, for this case, the only choice for $V$ is $B^2$.
\end{proof}

It is important to note that  this theorem determines the spatial topology of the DOC, the 
region of  spacetime which can be probed from infinity, not that of the full spacetime. Thus
it  allows the Cauchy slice of the full spacetime $M^{2+1}$ to have non-trivial topology; this topology, however, will be behind a single horizon. This result is shown below.

\begin{thm}\label{topcenhor}  Let $M^{2+1}$ be a globally hyperbolic AF or ALADS geon spacetime  satisfying NEC, 
and let $V$ be a Cauchy surface in $M^{2+1}$. Then the spacetime has a horizon.

\end{thm} 

\begin{proof} By the definition of a geon, the interior of the Cauchy surface is $\Pi - \{p\}$. Assume that the geon spacetime has no horizon; it follows that the DOC of $M^{2+1}$ is itself  $M^{2+1}$.
However, according to Theorem \ref{topcen}, the interior of the Cauchy surface of $M^{2+1}$ is then either $B^2$ or $I\times S^1$. These spaces are topologically $S^2 - \{p\}$ or  ${\mathbb{R}}^2 - \{p\}$. Hence neither contain any nontrivial topology and are not the Cauchy 
surface of a geon, in contradiction to the initial assumption.  Therefore the geon spacetime has a horizon.
\end{proof}

Theorems \ref{topcen} and  \ref{topcenhor} imply that if a spacetime has any non-trivial topology, it is hidden inside a
horizon. Moreover, there is never more than one black hole horizon bounding the DOC of each connected component of $\scri$ in any $2+1$-dimensional AF or ALADS spacetime. Hence an observer in the asymptotic region characterizes the spacetime as containing a single black hole.\footnote{ In 2+1 gravity, Brill \cite{Brill:1995jv} uses the terminology multi-black hole spacetime to describe the structure certain 2+1 ALADS spacetimes. This terminology does not mean multiple black hole horizons in a single asymptotic region. Multi-black hole spacetimes, as clearly noted by Brill, in fact have multiple disconnected components of $\scri$ with a single
black hole horizon bounding the DOC of each disconnected component.  Thus they clearly satisfy our theorems.}
This still leaves open the question of whether or not there exist geon solutions of the Einstein equations that exhibit horizons. The answer is clearly yes for the ALADS case; the spinning anti-de Sitter wormhole \cite{Aminneborg:1998si} is such a geon. Another easily exhibited geon is the $RP^2$ geon, the 3-dimensional analog of the $RP^3$ geon of \cite{Friedman:1993ty}. It can be built from the standard $2+1$-dimensional AdS black hole solution  \cite{Banados:1992wn} with spatial topology $S^1\times \mathbb R$ by identification of the two asymptotic regions to form a single one. It contains a $t=0$ hypersurface formed from that of the AdS black hole by identification of antipodal points on its totally geodesic $S^1$ throat. This hypersurface consequently has the topology of $RP^2 - \{ p\}$. Other ALADS solutions, including those whose Cauchy surfaces have nontrivial topology and have more than one asymptotic region, are also  examples with horizons \cite{Brill:1995jv,Aminneborg:1997pz}. 

Therefore the key question is whether or not there exist geon solutions for the case of zero cosmological constant. To answer this equation we first establish that all AF geons contain an outer trapped surface.

\begin{thm}\label{nooutertrapped}   
Every globally hyperbolic, AF geon that satisfies the NEC has a cover whose Cauchy surface contains an outer trapped surface.
\end{thm}

\begin{proof}
Assume that there exists $M^{2+1}$, an globally hyperbolic AF geon.
First, the DEC implies the NEC; therefore Theorem \ref{topcenhor}
implies $M^{2+1}$ must have all topology behind a horizon. We now show a contradiction to the existence of a horizon for the AF case. 
We do so by showing that the existence of a horizon implies that there is an outer trapped surface in a AF spacetime that is a cover of the geon. Let the Cauchy surface of $M^{2+1}$ be $V$. As $M^{2+1}$ is globally hyperbolic , $\pi _1 (M^{2+1})=\pi _1(V)$. 
As $M^{2+1}$ is a geon,  $\pi _1(V)\neq 1$. By definition, $V=\Pi-\{p\}$
where $\Pi$ is a closed 2-manifold other than $S^2$. Therefore
 $\Pi$ must be multiply-connected, and as such has a universal covering space  $\hat\Pi$.
 
Let $\hat{\{ p \}}$ be the inverse image of the point $\{p\}$ in $\hat\Pi$. Removing  the  points $\hat{\{ p \}}$  from $\hat{\Pi}$ gives us a space $\hat V$. $\hat V$ is a covering space of $V$ and has multiple disconnected asymptotic regions. Let  $\hat M^{2+1}$ be the corresponding covering space of  $M^{2+1}$ induced by the multiple 
covering of the Cauchy slice $V$.  The geometries of $\hat V$ and $\hat M^{2+1}$ are locally the same as 
$V$ and $M^{2+1}$, respectively. Therefore $\hat M^{2+1}$ will have multiple disconnected asymptotic regions each of which are AF.

Let $\scri_0$ be a connected component of the boundary at infinity of $\hat M^{2+1}$ corresponding to
the  $\scri$ of $M^{2+1}$. By Theorem \ref{topcenhor}, the other multiple disconnected asymptotic regions are separated from $\scri_0$ by horizons. Let $\scri_i$ be the boundary at infinity of one such disconnected region.  As the spacetime is AF, one can find large spacelike circles in the neighborhood of $\scri_i$ that have  expansion $\theta$ with the sign of that in Minkowski spacetime; positive outward and negative inward where outward is defined as the radial spatial direction toward $\scri_i$.  This surface is outer trapped with respect to $\scri_0$ as any curve from this surface to $\scri_0$ must travel radially inward, that is away from $\scri_i$.  As the expansion in this direction is negative, this large circle is outer trapped with respect  to $\scri_0$. As the Cauchy surface can be chosen such that it contains this large circle, there is an outer trapped surface in $\hat V$.\footnote{Observe that the construction in \cite{Wald:1991} shows a special slicing of extended Schwarzschild
that does not exhibit outer trapped surfaces in the black hole region of $r\leq 2M$.   These Cauchy surfaces clearly contain outer trapped surfaces with respect to $\scri_0$; they have $r>2M$ and lie in the second asymptotic region. Their existence is irrelevant for the purpose of the Wald and Iyer example as black holes formed from collapse do not have a second asymptotic region.  Furthermore, other slicings of Schwarzschild have outer trapped surfaces in the black hole region.}\end{proof}

Also note that Cauchy surfaces containing outer trapped surfaces must also exist for ALADS spacetimes in 3 dimensions by a similar argument.

It is important to note that this proof does not use the universal cover of $V$. In 3-dimensional AF spacetimes,
each connected component of  $\scri $ is not itself simply connected,  in contrast to the case in 4 or more spacetime dimensions. This fact means that the  construction of the
universal cover of $V$ unwinds not only nontrivial curves that characterize  the geon itself but also those that wrap $\scri$. This results in  
an unnecessarily complicated covering space with very different structure than that needed to prove  Theorem \ref{nooutertrapped}  and its corollaries.

Given that the existence of a horizon in a AF geon spacetime implies the existence of an outer trapped surface, we can prove our final results.

\begin{cor}  \label{nogeon2}
There are no globally hyperbolic, AF geons $M^{2+1}$ that satisfy the vacuum Einstein equations.
\end{cor}
\begin{proof}
If such an $M^{2+1}$ exists, Theorems \ref{topcenhor} and \ref{nooutertrapped} imply that it has a horizon and its cover has an outer trapped surface. Theorem \ref{trapsur} and its Corollary \ref{noanalytic} apply as vacuum AF spacetimes in 3 dimensions are analytic by Theorem \ref{analy}. Thus no Cauchy surface of the covering space $\hat M^{2+1}$ of the AF geon can contain an outer trapped surface. Hence, there is a contradiction to Theorem \ref{nooutertrapped}. Thus such a geon cannot exist.
\end{proof}

This corollary can readily be generalized to other cases of interest;
\begin{cor}  
There are no globally hyperbolic, AF geons that satisfy the vacuum Einstein-Maxwell equations.
\end{cor}
Its proof follows directly by a similar proof to Corollary \ref{nogeon2}; in this case the presence of the Maxwell stress energy tensor implies that $c \neq 0$ so the result follows directly from Theorem \ref{trapsur}. At this point it is clear that the proof of the non-existence of AF geons can be extended to any case for which there can be no outer trapped surface. Such cases include all geons with long range matter fields.

\section{Discussion} 

We have demonstrated that the existence of $2+1$-dimensional  geons
requires a negative cosmological constant.  Moreover, Theorem \ref{topcenhor} 
implies the topology any spacetime  is always hidden behind a single horizon. Specifically,  if one has two or more geons in anti-de Sitter spacetime consisting  of a Cauchy surface containing two or more spatially separated local topological structures such as two or more handles on ${\mathbb R}^2$, an  observer in the asymptotic region only sees a single horizon that hides all
structures from their view.  This is in marked contrast to the 4-dimensional case. 

Although our proof of Theorem \ref{trapsur} is essentially the same as that of \cite{Ida:2000jh}, we have paid more attention to the case of zero cosmological constant. In particular, we have used analyticity to rule out the case of vacuum geons. This step is essential to our work. The fact that the argument for the analyticity of the vacuum Einstein equations in 3 dimensions can be made without assumption of stationarity or staticity is again in contrast to the 4-dimensional case.

Our careful treatment of the vacuum case is also essential to any rigorous proof that there is no Schwarzschild solution in $2+1$ gravity. Theorem \ref{nooutertrapped} and the observation that existence of a Schwarzschild solution implies existence of an $RP^2$ geon also implies that there is no Schwarzschild solution due to this careful analysis of the vacuum case. We are not aware of this appearing previously  in the literature.
Although we have concentrated attention on the vacuum and Einstein-Maxwell cases, 
 similar corollaries clearly can be proven for any other matter fields coupled to gravity that satisfy the DEC. Furthermore, it is reasonable to anticipate that the rigidity of $2+1$ gravity allows one to rule out the existence of outer trapped surfaces in nonanalytic solutions as well. 
 
 Finally, note that these results apply not just to Einstein gravity but to any 2+1 spacetime that satisfies the conditions of the theorems. In particular, note that the null energy condition is in fact a condition on the Ricci curvature and can be easily checked for any spacetime. Similarly, the dominant energy condition is used in Theorem \ref{trapsur} to obtain a definite sign on the curvature contraction in (\ref{curve}). Therefore, these conditions can be checked for solutions of any 2+1 gravitational theory, including higher derivative theories. 
 
Interestingly, certain common 2+1 gravitational theories violate the needed condions. For example, solutions of the gravitating nonlinear $O(3)$ $\sigma$ model \cite{Clement:1998an} satisfies the NEC for positive gravitational constant and violates the NEC for negative gravitational constant. Not surprisingly, these solutions also violate the conditions needed on the curvature for Theorem \ref{trapsur} for the negative gravitational constant case. It is easy to see why in this model; changing the sign of the gravitational constant changes the sign of the stress energy tensor. Consequently, the DEC is violated.  The violation of the NEC and DEC are why the $O(3)$ model exhibits multiple black hole solutions.
 
Secondly, certain nontrivial solutions of topologically massive gravity violate NEC.  The MCL black hole \cite{Moussa:2003fc}, a nontrivial vacuum solution to topologically massive gravity, violates the NEC.  
It  also violates the curvature conditions needed for Theorem \ref{trapsur}, as one would expect. Similarly,  warped AdS black holes in 2+1 dimensions \cite{Nutku:1993eb,Bouchareb:2007yx,Anninos:2008fx}  also generically violate the NEC. The spatially warped AdS black hole violates NEC unless $\mu l = 3$, in which case it  is simply AdS$_3$.  The timelike warped AdS black hole violates the NEC in the squashed case,  $\mu l < 3$. \footnote{The stretched case contains closed timelike curves and hence is not globally hyperbolic.} Therefore, these examples violate the conditions needed for the topological censorship theorem proven in Section \ref{sec4}. Interestingly, this violation is independent of the sign of the gravitational constant. Therefore multiple black hole solutions could, in principle, exist in these theories. Thus geon solutions in which the geons are behind separate black hole horizons could also exist in these theories.  Finally, as the NEC is used in the proof of the singularity theorems, there is no reason that geon solutions need be  shrouded by black hole horizons. 

Note that satisfaction of the NEC is a crucial ingredient for fundamental results in general relativity. The NEC or a convergence condition stronger than the NEC is imposed in all version of the singularity theorems \cite{he}. It is also imposed in the area theorem of Hawking. Consequently, in solutions that violate the NEC, one may not have the connection between increasing horizon area and increase of entropy as dictated by the second law of thermodynamics. 

In summary, there are no geon containing solutions of asymptotically flat vacuum gravity in 2+1 dimensions. Solutions to asymptotically AdS 2+1 gravity can contain geons; however all geons are hidden behind a single black hole horizon. In contrast, solutions in theories such as the $O(3)$ $\sigma$ model with negative gravitational constant and topologically massive gravity potentially can contain geons, either each behind a separate black hole horizon or possibly behind no horizons as solutions in these theories violate the NEC. However, as the NEC is violated, one expects these solutions to be generically unstable. This instability may be related to those noted in topologically massive gravity at the chiral point from a perturbative viewpoint \cite{Grumiller:2008qz}.

\acknowledgments The work was supported by NSERC. In addition, the authors Schleich and Witt 
would like to thank the Perimeter Institute for its hospitality during the period when the ideas 
for this paper were inspired. We would also like to thank Greg Galloway for his help in providing reference \cite{Galloway:1994a} and discussions and Daniel Grumiller for questions on the case of topologically massive gravity.


\begin{thebibliography}{99}

\bibitem{Wheeler:1955zz}
  J.~A.~Wheeler,
  Phys.\ Rev.\  {\bf 97}, 511 (1955).
  
\bibitem{Wheeler:1957mu}
  J.~A.~Wheeler,
  Annals Phys.\  {\bf 2}, 604 (1957).
  
\bibitem{Misner:1957mt}
  C.~W.~Misner and J.~A.~Wheeler,
  Annals Phys.\  {\bf 2}, 525 (1957).
  
\bibitem{Brill:1957}
D.~R.~Brill and J.~A.~Wheeler, 
Rev.\ Mod.\ Phys.\ {\bf  29}, 465 (1957).



\bibitem{Ernst:1957}
F.~J.~Ernst, Jr., 
Phys.\  Rev. \ {\bf 105}, 1665 (1957).

\bibitem{Brill:1964}
D.~R.~Brill and J.~B.~Hartle,
 Phys.\ Rev.\ {\bf 135}, B271 (1964).



\bibitem{Sorkin:1979ja}
  R.~Sorkin,
  J.\ Phys.\ A  {\bf 12}, 403 (1979).
  
\bibitem{Friedman:1980st}
  J.~L.~Friedman and R.~D.~Sorkin,
  Phys.\ Rev.\ Lett.\  {\bf 44}, 1100 (1980).
  
  
\bibitem{Friedman:1982du}
  J.~L.~Friedman and R.~D.~Sorkin,
  Gen.\ Rel.\ Grav.\  {\bf 14}, 615 (1982).


  
\bibitem{Sorkin:1985bg}
  R.~D.~Sorkin,
  in{\it Topological Properties and Global Structure of Spacetime: Proceedings},
  NATO Advanced Study Institute Series B: Physics v. 138, eds. Bergmann and De Sabbata, (Plenum, New York, 1986).
  
  
  
  
\bibitem{Friedman:1983ft}
  J.~l.~Friedman and D.~M.~Witt,
  Phys.\ Lett.\  B {\bf 120}, 324 (1983).
  
\bibitem{Witt:1986ef}
  D.~M.~Witt,
  J.\ Math.\ Phys.\  {\bf 27}, 573 (1986).
  

  
\bibitem{Friedman:1986ze}
  J.~L.~Friedman and D.~M.~Witt,
  Topology {\bf 25}, 35 (1986).
  
\bibitem{Friedman:1988he}
  J.~L.~Friedman and D.~M.~Witt,
  Contemp.\ Math.\  {\bf 71}, 301 (1988).
  
   
\bibitem{Witt:1986ng}
  D.~M.~Witt,
  Phys.\ Rev.\ Lett.\  {\bf 57}, 1386 (1986).

\bibitem{Aneziris:1988xz}
  C.~Aneziris, A.~P.~Balachandran, M.~Bourdeau, S.~Jo, T.~R.~Ramadas and R.~D.~Sorkin,
  Mod.\ Phys.\ Lett.\  A {\bf 4}, 331 (1989).

\bibitem{Aneziris:1989cr}
  C.~Aneziris, A.~P.~Balachandran, M.~Bourdeau, S.~Jo, T.~R.~Ramadas and R.~D.~Sorkin,
  Int.\ J.\ Mod.\ Phys.\  A {\bf 4}, 5459 (1989).

\bibitem{Balachandran:1990wr}
  A.~P.~Balachandran, A.~Daughton, Z.~C.~Gu, G.~Marmo, R.~D.~Sorkin and A.~M.~Srivastava,
  Mod.\ Phys.\ Lett.\  A {\bf 5}, 1575 (1990)
  [Phys.\ Scripta {\bf T36}, 253 (1991)].

\bibitem{Sorkin:1996yt}
  R.~D.~Sorkin and S.~Surya,
  Int.\ J.\ Mod.\ Phys.\  A {\bf 13}, 3749 (1998)
  [arXiv:gr-qc/9605050].

\bibitem{Dowker:2000zy}
  H.~F.~Dowker and R.~D.~Sorkin,
  AIP Conf.\ Proc.\  {\bf 545}, 205 (2004)
  [arXiv:gr-qc/0101042].

\bibitem{Witten:1988hc}
  E.~Witten,
  Nucl.\ Phys.\  B {\bf 311}, 46 (1988).

\bibitem{Carlip:1998uc}
  S.~Carlip,
{\it  Cambridge, UK: Univ. Pr. (1998) 276 p}

\bibitem{Carlip:2004ba}
  S.~Carlip,
  Living Rev.\ Rel.\  {\bf 8}, 1 (2005)
  [arXiv:gr-qc/0409039].
  
  

\bibitem{Samuel:1993ua}
  J.~Samuel,
  Phys.\ Rev.\ Lett.\  {\bf 71}, 215 (1993).

\bibitem{Balachandran:2000kv}
  A.~P.~Balachandran, E.~Batista, I.~P.~Costa e Silva and P.~Teotonio-Sobrinho,
  Mod.\ Phys.\ Lett.\  A {\bf 16}, 1335 (2001)
  [arXiv:hep-th/0005286].



\bibitem{Witten:2007kt}
  E.~Witten,
  arXiv:0706.3359 [hep-th].
 


\bibitem{Carlip:2008jk}
  S.~Carlip, S.~Deser, A.~Waldron and D.~K.~Wise,
  arXiv:0803.3998 [hep-th].
  
\bibitem{Carlip:2008eq}
  S.~Carlip, S.~Deser, A.~Waldron and D.~K.~Wise,
  Phys.\ Lett.\  B {\bf 666}, 272 (2008)
  [arXiv:0807.0486 [hep-th]].
  

  \bibitem{Li:2008dq}
  W.~Li, W.~Song and A.~Strominger,
  JHEP {\bf 0804}, 082 (2008)
  [arXiv:0801.4566 [hep-th]].
  

    

\bibitem{Brill:1995jv}
  D.~R.~Brill,
  Phys.\ Rev.\  D {\bf 53}, 4133 (1996)
  [arXiv:gr-qc/9511022].

\bibitem{Aminneborg:1997pz}
  S.~Aminneborg, I.~Bengtsson, D.~Brill, S.~Holst and P.~Peldan,
  Class.\ Quant.\ Grav.\  {\bf 15}, 627 (1998)
  [arXiv:gr-qc/9707036].
  
\bibitem{Aminneborg:1998si}
  S.~Aminneborg, I.~Bengtsson and S.~Holst,
  Class.\ Quant.\ Grav.\  {\bf 16}, 363 (1999)
  [arXiv:gr-qc/9805028].
  
  
\bibitem{Ida:2000jh}
  D.~Ida,
  Phys.\ Rev.\ Lett.\  {\bf 85}, 3758 (2000)
  [arXiv:gr-qc/0005129].
  


\bibitem{Friedman:1993ty}
J.~L.~Friedman, K.~Schleich and D.~M.~Witt,
Phys.\ Rev.\ Lett.\  {\bf 71}, 1486 (1993)
[Erratum-ibid.\  {\bf 75}, 1872 (1995)]
[arXiv:gr-qc/9305017].

\bibitem{Galloway:1999bp}
G.~J.~Galloway, K.~Schleich, D.~M.~Witt and E.~Woolgar,
Phys.\ Rev.\ D {\bf 60}, 104039 (1999)
[arXiv:gr-qc/9902061].

\bibitem{Galloway:1999br}
G.~J.~Galloway, K.~Schleich, D.~Witt and E.~Woolgar,
Phys.\ Lett.\ B {\bf 505}, 255 (2001)
[arXiv:hep-th/9912119].


\bibitem{Ashtekar:2002qc}
  A.~Ashtekar, J.~Wisniewski and O.~Dreyer,
  Adv.\ Theor.\ Math.\ Phys.\  {\bf 6}, 507 (2003)
  [arXiv:gr-qc/0206024].
  
  
  \bibitem{Hawking:1973}
S.~W.~Hawking, in {\it Black Holes}, eds. C. DeWitt and B. DeWitt (Gordon and Breach,
New York, 1973), 1.

\bibitem{Hawking:1971vc}
  S.~W.~Hawking,
  Commun.\ Math.\ Phys.\  {\bf 25}, 152 (1972).

\bibitem{he} S.W. Hawking\&  G.F.R. Ellis, {\it The large scale
structure of space-time}, (Cambridge University Press, Cambridge, 1973).


\bibitem{Woolgar:1999yi}
  E.~Woolgar,
  Class.\ Quant.\ Grav.\  {\bf 16}, 3005 (1999)
  [arXiv:gr-qc/9906096].
  
\bibitem{Galloway:1994a}
  G.~J.~Galloway,
  Contemp.\  
Math.\  {\bf 170}, 113 (1994) (eds. J. K. Beem and K. Duggal).

\bibitem{Banados:1992wn}
  M.~Banados, C.~Teitelboim and J.~Zanelli,
  Phys.\ Rev.\ Lett.\  {\bf 69}, 1849 (1992)
  [arXiv:hep-th/9204099].
 

  
\bibitem{Muller:1970}
 H.~M$\ddot{\textrm{u}}$ller Zum Hagen,
 Proc.\ Camb.\ Phil.\ Soc\ {\bf 67}, 415 (1970).

\bibitem{Muller:1971}
 H.~M$\ddot{\textrm{u}}$ller Zum Hagen,
 Proc.\ Camb.\ Phil.\ Soc\ {\bf 68}, 199 (1970).
  
\bibitem{Tod:2007tb}
  P.~Tod,
  Gen.\ Rel.\ Grav.\  {\bf 39}, 1031 (2007)
  [arXiv:0704.2508 [gr-qc]].
  



\bibitem{wolf} J. A. Wolf, {\it Spaces of Constant Curvature}, (Publish or Perish, Wilmington, 
Delaware (USA), 1984).

\bibitem{MorrowJones:1988yw}
  J.~W.~Morrow-Jones,
  ``Nonlinear Theories of Gravity: Solutions, Symmetry and Stability",
  citation = UMI-89-05295;

\bibitem{MorrowJones:1993zu}
  J.~Morrow-Jones and D.~M.~Witt,
  Phys.\ Rev.\  D {\bf 48}, 2516 (1993).

\bibitem{mess}
G.~Mess,`` Lorentz spacetimes of constant curvature", MSRI Preprint 90-05808, 1990.

 \bibitem{Wald:1991}
 R.~M.~Wald and V.~Iyer,
 Phys.\ Rev.\ D\ {\bf 44}, R3719  (1991) 

\bibitem{Clement:1998an}
  G.~Clement and A.~Fabbri,
  Class.\ Quant.\ Grav.\  {\bf 16}, 323 (1999)
  [arXiv:gr-qc/9804050].
  
\bibitem{Moussa:2003fc}
  K.~A.~Moussa, G.~Clement and C.~Leygnac,
  Class.\ Quant.\ Grav.\  {\bf 20}, L277 (2003)
  [arXiv:gr-qc/0303042].


\bibitem{Anninos:2008fx}
  D.~Anninos, W.~Li, M.~Padi, W.~Song and A.~Strominger,
  arXiv:0807.3040 [hep-th].
  
\bibitem{Bouchareb:2007yx}
  A.~Bouchareb and G.~Clement,
  Class.\ Quant.\ Grav.\  {\bf 24}, 5581 (2007)
  [arXiv:0706.0263 [gr-qc]].
  
\bibitem{Nutku:1993eb}
  Y.~Nutku,
  Class.\ Quant.\ Grav.\  {\bf 10}, 2657 (1993).
  
\bibitem{Grumiller:2008qz}
  D.~Grumiller and N.~Johansson,
  JHEP {\bf 0807}, 134 (2008)
  [arXiv:0805.2610 [hep-th]].

  


\end{thebibliography}
\end{document}